\documentclass[twocolumn]{autart}    % Enable this line and disable the 
\usepackage{graphicx}
\usepackage{amsmath} 
\usepackage{amsthm}
\usepackage{amssymb} 
\usepackage{algorithm} 
\usepackage{algorithmic}
\usepackage{mathtools}
\usepackage[colorlinks=true,linkcolor=blue,citecolor=blue,urlcolor=blue]{hyperref}
\usepackage{url,cite} 
\usepackage{tikz}

\newcommand{\by}{\mathbf{y}}
\newcommand{\bY}{\mathbf{Y}}
\newcommand{\bX}{\mathbf{X}}

\newtheorem{theorem}{Theorem}
\newtheorem{proposition}{Proposition}
\newtheorem{remark}{Remark}

\newcommand{\rank}{\mathop{\text{\normalfont rank}}}
\newcommand{\spn}{\mathop{\text{\normalfont span}}}

\newcommand{\st}{\mathop{\text{\normalfont s.t.}}}
\newcommand{\row}{\mathop{\text{\normalfont row}}}

\begin{document}

\begin{frontmatter}

\title{Unifying Theorems for Subspace Identification \\
and Dynamic Mode Decomposition}% \thanksref{footnoteinfo}} % Title, preferably not more 
\author[Madison]{Sungho Shin}\ead{sungho.shin@wisc.edu},    % Add the 
\author[Madison]{Qiugang Lu}\ead{glu67@wisc.edu},               % e-mail address 
\author[Madison]{Victor M. Zavala}\ead{victor.zavala@wisc.edu}  % (ead) as shown

\address[Madison]{Department of Chemical and Biological Engineering, University of Wisconsin-Madison, Madison, WI 53706 USA} 
          
\begin{keyword}                      
  System identification, subspace methods, dynamic mode decomposition, optimization
\end{keyword}

\begin{abstract}
  This paper presents unifying results for subspace identification (SID) and dynamic mode decomposition (DMD) for autonomous dynamical systems. We observe that SID seeks to solve an optimization problem to estimate an extended observability matrix and a state sequence that minimizes the prediction error for the state-space model. Moreover, we observe that DMD seeks to solve a rank-constrained matrix regression problem that minimizes the prediction error of an extended autoregressive model. We prove that existence conditions for perfect (error-free) state-space and low-rank extended autoregressive models are equivalent and that the SID and DMD optimization problems are equivalent. We exploit these results to propose a SID-DMD algorithm that delivers a provably optimal model and that is easy to implement. We demonstrate our developments using a case study that aims to build dynamical models directly from video data.
\end{abstract}

\end{frontmatter}

\section{Introduction}

This paper considers the problem of identifying an autonomous dynamical model from a sequence of output (observable) data by using subspace identification (SID) and dynamic mode decomposition (DMD). 

SID seeks to identify a dynamical model in state-space form from output sequence data \cite{van2012subspace,qin2006overview}. Since state data are not available (the states are unknown), SID adopts a sequential approach wherein an extended observability matrix and the state sequence are first identified from the output data and these quantities are then used to identify the state-space system matrices. The term {\it subspace} arises from the fact that the state sequence is identified from a subspace defined by a delay-embedded output sequence \cite{van2012subspace}. A number of different SID algorithms have been proposed in the literature such as PC, UPC, and CVA \cite{arun1990balanced} (for autonomous systems) and N4SID \cite{van1994n4sid}, MOSEP \cite{verhaegen1994identification}, and CVA \cite{larimore1990canonical} (for non-autonomous systems). In seminal work, Van Overschee and De Moor established a {\em unifying theorem}, which indicates that the only defining feature of such methods is the weighting scheme used for singular value decomposition (SVD) \cite{van1995unifying}. SID has been the dominant paradigm for state-space dynamic model identification in industrial applications \cite{favoreel2000subspace}. Interestingly, optimality properties for SID models have not been explored in the literature; specifically, existing SID approaches are often justified based on geometric interpretations (e.g.,  \cite[Chapter 2]{van2012subspace}) and not on optimization/regression interpretations.  

DMD is an identification method that has recently gained considerable attention in the literature as it provides a scalable approach to deal with high-dimensional state spaces (as those arising in computational fluid dynamics and partial differential equations)  \cite{schmid2010dynamic,tu2013dynamic,kutz2016dynamic}. In DMD, a low-rank dynamical model is obtained by solving a linear regression problem (which can be solved by using SVD) and the rank is constrained via truncation of SVD. In recent work, an alternative rank-constrained regression formulation for DMD is proposed and a closed-form solution for such a problem was derived  \cite{heas2016low}. This result implies that external truncation only delivers a suboptimal solution. The models resulting from DMD are often expressed in terms of {\em modes},  which represent coherent structures and associated timescales (this information is of high value in understanding spatiotemporal phenomena). Connections between DMD and Koopman operator theory have also been established in  \cite{lasota2013chaos,schmid2010dynamic,tu2013dynamic}. Unlike SID, however, the model order of DMD is assumed to be equal to the order of a delay-embedded observable (i.e., low-order state space is not formally constructed). As a result, DMD models do not provide a direct low-order state-space representation (unless the states are assumed equal to the outputs). Connections between SID and DMD have only been addressed superficially in \cite{proctor2016dynamic,arbabi2017ergodic}; to the best of our knowledge, no theoretical results have been established. 

This paper provides a theoretical treatment that unifies SID and DMD. We establish optimization problem formulations for SID and DMD. We show that existence conditions for perfect models (in the sense that they fit the data without errors) are equivalent and that both optimization problems are equivalent. We exploit the equivalence between SID and DMD to derive a method (that we call SID-DMD) that computes a state-space model from output data. Notably, the computed model has provable optimality properties (thus overcoming deficiencies of existing SID approaches) and is easy to implement (via SVD). Although SID and DMD algorithms can be used to identify non-autonomous systems, in this paper we aim to make a first attempt to unify these approaches by focusing on autonomous systems.

The remainder of the paper is organized as follows: In Section \ref{sec:unifying} we introduce optimization formulations for SID and DMD and equivalence theorems. In Section \ref{sec:alg} we propose the SID-DMD approach and properties of its solution. Numerical results for a case study using video data are presented in Section \ref{sec:num}.

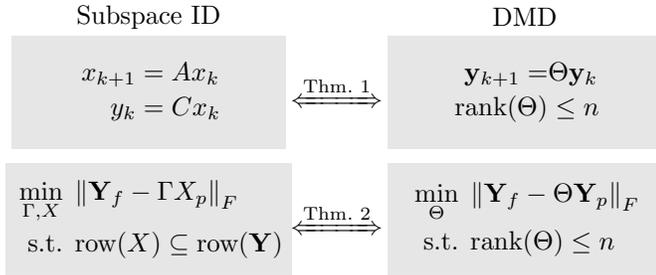
\begin{figure}[!htp]
    \centering
  \begin{tikzpicture}
    \node at (-2.5,0) {Subspace ID};
    \node at (2.5,0) {DMD};
    \node[fill,fill opacity=0.1,text opacity=1,minimum width=3.65cm,minimum height=1.5cm] at (-2.5,-1) {$ \begin{aligned}x_{k+1} &= Ax_k\\ y_{k} &= Cx_k\end{aligned}$};
    \node[fill,fill opacity=0.1,text opacity=1,minimum width=3.65cm,minimum height=1.5cm] at (2.5,-1) {$\begin{aligned}\by_{k+1} =& \Theta \by_{k}\\ \rank(\Theta)&\leq n\end{aligned}$};
    \node at (0,-1) {$\xLeftrightarrow{\;\text{Thm. \ref{thm:equiv-1}\;}}$};
    
    \node[fill,fill opacity=0.1,text opacity=1,minimum width=3.65cm,minimum height=1.5cm] at (-2.5,-2.7) { $ \begin{aligned}\min_{{\Gamma},X}\;& \left\Vert \bY_f - {\Gamma} X_p\right\Vert_F\\\st\;& \row(X) \subseteq \row(\bY)\end{aligned}$};
    \node[fill,fill opacity=0.1,text opacity=1,minimum width=3.65cm,minimum height=1.5cm] at (2.5,-2.7) { $\begin{aligned}\min_{\Theta}\;& \left\Vert \bY_f - \Theta  \bY_p \right\Vert_F\\\st\;&\rank(\Theta)\leq n\end{aligned}$};
    \node at (0,-2.7) {$\xLeftrightarrow{\;\text{Thm. \ref{thm:equiv-2}\;}}$};
  \end{tikzpicture}
  \caption{A schematic summary of unifying results.}\label{fig:unifying}
\end{figure}

\section{Unifying Results}\label{sec:unifying}

This section presents unifying theorems for SID and DMD; a summary of these results is shown in Fig. \ref{fig:unifying}.  We begin the discussion by defining some basic notation. The set of real numbers and integers are denoted by $\mathbb{R}$ and $\mathbb{I}$. By default, we consider vectors as column vectors and use syntax $[\xi_1;\cdots; \xi_n]=[ \xi_1^\top\;\cdots\; \xi_n^\top]^\top$. The submatrix of $ \xi$ with row indexes $i_1,\cdots,i_2$ and column indexes $j_1,\cdots,j_2$ is denoted by $ \xi[i_1:i_2,j_1:j_2]$. Moore-Penrose pseudoinverses are denoted by $(\cdot)^\dag$ and Frobenius norms are denoted by $\Vert \cdot \Vert_F$. We assume that a sequence of observable output data $\{y_k\in\mathbb{R}^m\}_{k=i}^j$ with $i,j\in\mathbb{I}\cup\{\pm\infty\}$ is available; if $i<j$ are finite, we can construct block-Hankel data matrices by embedding a time delay of order $s\in\mathbb{I}_{>0}$ as:
\begin{align*}
  \bY :=&
  \begin{bmatrix}
    y_{i}&y_{i+1}&\cdots&y_{j-s+1}\\
    y_{i+1}&y_{i+2}&\cdots&y_{j-s+2}\\
    \vdots&\vdots&\ddots&\vdots\\
    y_{i+s-1}&y_{i+s}&\cdots&y_{j}
  \end{bmatrix}
  =
  \begin{bmatrix}
    \by_i&\by_{i+1}&\cdots&\by_{j-s+1}
  \end{bmatrix}\\
  \bY_p&:=\bY[:,1:\ell]
  =\begin{bmatrix}
  \by_i&\by_{i+1}&\cdots&\by_{j-s},
  \end{bmatrix}\\
  \bY_f&:=\bY[:,2:\ell+1]
  =\begin{bmatrix}
  \by_{i+1}&\by_{i+2}&\cdots&\by_{j-s+1}
  \end{bmatrix}.
  %% \bY_p& :=
  %% \begin{bmatrix}
  %%   y_{i}&y_{i+1}&\cdots&y_{y-s}\\
  %%   y_{i+1}&y_{i+2}&\cdots&y_{y-s+1}\\
  %%   \vdots&\vdots&\ddots&\vdots\\
  %%   y_{i+s-1}&y_{i+s}&\cdots&y_{j-1}
  %% \end{bmatrix}\in\mathbb{R}^{ms\times \ell},
\end{align*}
where $\by_k:=[y_k;y_{k+1};\cdots;y_{k+s-1}]$; $\ell:=j-i-s+1$. 

\subsection{Subspace Identification}
SID aims to identify a state-space model of order $n\in\mathbb{I}_{>0}$ (user-defined) of the form:
\begin{subequations}\label{eqn:ss}
  \begin{align}
     x_{k+1} &= Ax_{k}+w_k\\
    y_{k} &= Cx_k+v_k,
  \end{align}
\end{subequations}
where $x_k\in\mathbb{R}^n$ is the unknown (hidden) state, $w_k\in\mathbb{R}^n$ is the state prediction error, $v_k\in\mathbb{R}^{m}$ is the output prediction error, and $A\in\mathbb{R}^{n\times n}$ and $C\in\mathbb{R}^{m\times n}$ are the system matrices.

Direct estimation of $(A,C)$ from output data is challenging because one must simultaneously estimate the state sequence. Doing this explicitly would require solving a nonconvex optimization problem \cite{mckelvey1995identification}. SID seeks to avoid this by indirectly estimating $(A,C)$; to see how this is done, we first observe that the output can be predicted using the state and extended observability matrix as:
\begin{align}\label{eqn:exob}
  \by_{k} = \Gamma_s x_k + e_k
\end{align}
where $\Gamma_s:=[C;CA;\cdots;CA^{s-1}]$ and $e_k\in\mathbb{R}^{ms}$ is the prediction error. Moreover, we observe that the state sequence can be constructed from the delay-embedded output sequence as: $x_{k} = A\Gamma_{s}^\dag \by_{k-1} + f_k$, where $f_k\in\mathbb{R}^{ms}$ is the prediction error. This allows us to assume that the sequence $\{x_{k+1}\}_{k=i}^{j-s+1}$ can be found from the subspace defined by the row space formed by $\{\by_{k}\}_{k=i}^{j-s+1}$. Accordingly, we observe that the SID problem can be cast as the optimization problem:
\begin{subequations}\label{eqn:sid}
  \begin{align}
    \min_{{\Gamma},X}\;& \left\Vert \bY_f - {\Gamma} X_p\right\Vert_F\\
    \st\;& \row(X) \subseteq \row(\bY),
  \end{align}
\end{subequations}
where $X:=[x_{i+1}\;\cdots\;x_{j-s+2}]$; $X_p:=X[:,1:\ell]$ and $\row(\cdot)$ represents the row space of a matrix. The existence of a solution follows from Proposition \ref{prop:n4sid} (presented later). As is well-known, SID problems are ill-poised;  observe that, if $(\Gamma^*,X^*)$ is a solution, $(\Gamma^* T^{-1},TX^*)$ is a solution for any nonsingular $T\in\mathbb{R}^{n\times n}$. Therefore, it suffices to obtain $(\Gamma,X)$ and subsequently $(A,C)$ up to a similarity transformation.

\begin{remark}
Problem \eqref{eqn:sid} aims to minimize the prediction error for \eqref{eqn:exob} while enforcing the state sequence to lie in the subspace defined by the row space of $\bY$. It is important to note that existing SID algorithms have effectively solved \eqref{eqn:sid} (we show this in Appendix \ref{apx:sid}), but the optimization problem \eqref{eqn:sid} has not been stated explicitly in the literature (to the best of our knowledge). Existing SID approaches are often justified based on geometric interpretations (e.g., see \cite[Chapter 2]{van2012subspace}) and not on optimization/regression interpretations. 
\end{remark}

\begin{remark}
In conventional SID algorithms, the data matrices $\bY_p$ and $\bY_f$ are constructed in a way that the data in the corresponding columns are not overlapping with each other (e.g., by choosing $\bY_p:=\bY[1:\ell-s+1]$ and $\bY_f:=[s+1:\ell+1]$). Here, we have modified this by allowing data overlaps in order to ensure consistency with DMD models.  Moreover, in conventional methods, $(A,C)$ are identified by using the estimated $({\Gamma},X)$; for instance, one can use $A:=\Gamma[1:m(s-1),:]^\dag \Gamma[m+1:ms,:]$ and $C:=\Gamma[1:m,:]$. The work in \cite{van2012subspace} provides an overview of other approaches that use the state sequence $X$.
\end{remark}

\subsection{Dynamic Mode Decomposition}
DMD aims to identify a dynamic mapping (shift) between past and future states. When the full state is not observed, models are constructed by using a time-delay embedding \cite{kutz2016dynamic,arbabi2017ergodic}. Such a technique is related to the classical Takens theorem. By embedding a time-delay with order $s$, the model can be represented as:
\begin{align}\label{eqn:exar}
  \by_{k+1} = \Theta \by_{k} + g_k,
\end{align}
where $g_k\in\mathbb{R}^{ms}$ is the error. We call \eqref{eqn:exar} an {\it extended autoregressive} (AR) model, since each block row takes the form of an AR model. In typical applications such as fluid dynamics, the output data is high-dimensional and it is thus desired to induce simplicity in the identified model. This is done by constraining the rank of the mapping $\Theta$; specifically, we enforce $\rank(\Theta)\leq n$ for a given $n\in\mathbb{I}_{>0}$. We can thus see that DMD seeks to find a mapping that solves the optimization problem:
\begin{subequations}\label{eqn:lrank}
  \begin{align}
    \min_{\Theta}\;& \left\Vert \bY_f - \Theta  \bY_p \right\Vert_F\label{eqn:lrank-obj}\\
    \st\;&\rank(\Theta)\leq n.\label{eqn:lrank-con}
  \end{align}
\end{subequations}
This problem finds the mapping that minimizes the prediction error of \eqref{eqn:exar} (measured by the Frobenius norm) while satisfying the rank constraint of the mapping. The existence of a solution to this problem follows from Proposition \ref{prop:lrank}(a) (presented later). \\

\begin{remark} DMD was originally proposed in the literature by assuming a form much simpler form than \eqref{eqn:lrank} \cite{schmid2010dynamic}. Specifically, the notion of time-delay embedding was not formally introduced (states were assumed to be observable) and the rank condition for $\Theta$ was not directly enforced as constraints. Time-delay embedding was recently introduced in \cite{kutz2016dynamic} and in the Hankel-DMD framework of \cite{arbabi2017ergodic}. To enforce low rank, truncated SVD has been typically performed on the data matrix $\bY_p$  \cite{kutz2016dynamic}. Note that this approach delivers a suboptimal solution to \eqref{eqn:lrank}. To find an optimal low-rank mapping $\Theta$, one needs to directly find the solution of this problem. In recent work, DMD was formulated as a rank-constrained regression problem and a closed-form solution was derived \cite{heas2017optimal}. 
\end{remark}

\subsection{Equivalence Theorems}

We now present equivalence theorems for SID and DMD. The first theorem states that underlying model assumptions of SID and DMD are equivalent (the existence of a model that perfectly fits \eqref{eqn:ss} is equivalent to the existence of a model that perfectly fits \eqref{eqn:exar}). The second theorem states that the optimization problems of SID and DMD are equivalent (a solution of \eqref{eqn:sid} can be obtained from a solution of \eqref{eqn:lrank} and viceversa). 
\\

\begin{theorem}\label{thm:equiv-1}
  Given $\{y_k\}_{k=-\infty}^{\infty}$ and $n,s\in\mathbb{I}_{>0}$, the following statements are equivalent.\\
  (a) There exists $\{x_k\in\mathbb{R}^n\}_{k=-\infty}^{\infty}$ and $(A, C)$ with observability index not greater than $s$ such that \eqref{eqn:ss} is satisfied with $w_k=0$ and $v_k=0$ for $k\in\mathbb{I}$.\\
  (b) There exists $\Theta\in\mathbb{R}^{ms \times m s}$ with $\rank(\Theta)\leq n$ such that \eqref{eqn:exar} is satisfied with $g_k=0$  for $k\in\mathbb{I}$.
\end{theorem}

\begin{proof}[Proof of (a)$\implies$(b)] From the assumption that $(A, C)$ is observable with index not greater than $s$, we have that $\Gamma_s$ has full column rank. Observe now from (a) that $\by_k = \Gamma_sx_k$ holds. By left multiplying $\Gamma_s^\dag$, we have $ x_{k} =\Gamma_s^\dag  \by_k$; moreover, $ \by_{k+1} =  \Gamma_{s} Ax_{k}=\Gamma_{s}A \Gamma_s^\dag \by_k$. Thus, we can see that, for $\Theta :=  \Gamma_{s}A \Gamma_s^\dag$, \eqref{eqn:exar} holds with $g_k=0$. Since $\rank(\Gamma_s)\leq n$, $\rank(\Theta)\leq n$; thus (b) holds.
\end{proof}

\begin{proof}[Proof of (b)$\implies$(a)] We redefine $\Theta\leftarrow G G^\top \Theta$ where the columns of $ G$ form an orthonormal basis of $\spn\{ \by_k\}$ (here, one can easily show that span of an infinite vector set is a vector space, so the basis is well-defined). We observe that \eqref{eqn:exar} with $g_k=0$ and $\rank(\Theta)\leq n$ still hold; thus (b) is not violated by the redefinition of $\Theta$. We let $\Theta =  P Q^\top$, where $ P\in\mathbb{R}^{ms\times r}$, $Q\in\mathbb{R}^{ms\times r}$, and $r:=\rank(\Theta)$ (such a factorization always exists). Also, we let $ \widetilde{x}_{k+1}:= Q^\top  \by_k$, $\widetilde{A}:= Q^\top  P$, and $\widetilde{C}:= P[1:m,:]$. One can verify from \eqref{eqn:exar} that $\widetilde{x}_{k+1}=\widetilde{A}\widetilde{x}_k$ and $y_k=\widetilde{C}\widetilde{x}_k$ hold for $k\in\mathbb{I}$.
  Now observe that $\rank(\Theta)\leq \dim\{\by_k\}$ holds due to the projection operator $ GG^\top$, and $\rank(\Theta)\geq \dim\{\by_k\}$ holds due to $\{\by_k\}\subseteq \text{range}(\Theta)$, where $\dim(\cdot)$ denotes the dimension of the $\spn(\cdot)$. As such, $r = \dim\{ \by_k\}$. One can see that $\by_{k}=\widetilde{\Gamma}_s \widetilde{x}_k$ holds for any $k\in\mathbb{I}$ and $\dim\{\widetilde{x}_k\}=r$, and this implies that $\widetilde{\Gamma}_s$ has full column rank. We construct $\{x_k\}$ and $(A,C)$ as:
\begin{align*}\setlength\arraycolsep{2pt}
  x_k=\begin{bmatrix}
  \widetilde{x}_k\\
  0_{n-r}
  \end{bmatrix},\,
  A=
  \begin{bmatrix}
    \widetilde{A}& 0_{r\times n-r}\\
    0_{n-r\times r} & 0_{n-r\times n-r},
  \end{bmatrix},
  C=
  \begin{bmatrix}
    \widetilde{C}& C^{\perp}
  \end{bmatrix},
\end{align*}
where the columns of $C^{\perp}\in\mathbb{R}^{m\times (n-r)}$ are orthogonal to $C$. One can observe that \eqref{eqn:ss} holds with $w_k=0$ and $v_k=0$. From the fact that $\widetilde{\Gamma}_s$ has full column rank and $C^\perp$ is orthogonal to $C$, we have that $\Gamma_s$ has full column rank and thus (a) holds.
\end{proof}
Note that (a)$\implies$(b) is well-known but, to the best of our knowledge, (a)$\impliedby$(b) has not been proved before. \\

\begin{theorem}\label{thm:equiv-2}
  Let $(\Gamma^*,X^*)$ be a solution of \eqref{eqn:sid} and $\Theta^*=PQ^\top$ with $P,Q\in\mathbb{R}^{ms\times n}$ be a solution of \eqref{eqn:lrank}, then:\\
  (a) $\Theta = \Gamma^* X^*_p \bY_p^\dag$ is a solution of \eqref{eqn:lrank}.\\
  (b) $(\Gamma,X) = (P,Q^\top \bY)$ is a solution of \eqref{eqn:sid}.
\end{theorem}
\begin{proof}[Proof of (a)]
  Suppose that there exists $\Theta'=P'(Q')^\top$ with $P'\in\mathbb{R}^{ms\times n}$ and $Q'\in\mathbb{R}^{ms\times n}$, such that has better (lower) objective value for \eqref{eqn:lrank} than $\Gamma^*X^*_p\bY_p^\dag$. Observe:
  %\begin{align}\label{eqn:equiv-2-1}
  $\row((Q')^\top \bY)\subseteq\row(\bY)$.
  %\end{align}
  Furthermore, we have that $\row(X^*_p)\subseteq\row(\bY_p)$ because $\row(X^*)\subseteq\row(\bY)$. This implies that:
  %% \begin{align}\label{eqn:equiv-2-2}
  $X^*_p \bY_p^\dag \bY_p=X^*_p$.
  %% \end{align}
  Therefore, from the assumption that $\Theta'$ has better objective value for \eqref{eqn:lrank} than $\Gamma^* X_p^* \bY_p^\dag$, we have
  %% \begin{align}\label{eqn:equiv-2-3}
  $\Vert \bY_f - P'(Q')^\top \bY_p\Vert_F^2<\Vert \bY_f - \Gamma^* X^*_p\Vert_F^2$.
  %% \end{align}
  %% where the right-hand-side follows from \eqref{eqn:equiv-2-2}.
  Thus, now we know that $(P',(Q')^\top \bY) $ is feasible to \eqref{eqn:sid}, and it has better objective value for \eqref{eqn:sid} than $(\Gamma^*,X^*)$. This contradicts the assumption that $(\Gamma^*,X^*)$ is a solution of \eqref{eqn:sid}; therefore, (a) holds.
\end{proof}
\begin{proof}[Proof of (b)]
  Suppose there exists $(\Gamma',X')$ having a better objective value for \eqref{eqn:sid} than $(P,Q^\top \bY_p)$ and $\row(X')\subseteq \row(\bY)$. We observe from $\row(X')\subseteq \row(\bY)$ that there exists $T\in\mathbb{R}^{n\times ms}$ such that $X'=T\bY$ holds; thus, $\Gamma'X'_p= \Gamma' T \bY_p$. One can see from the dimension of $\Gamma'$ that $\rank(\Gamma' T)\leq n$ and from the assumption, we can see that $\Vert \bY_f - \Gamma'T\bY_p\Vert_F^2<\Vert \bY_f - PQ^\top \bY_p\Vert_F^2$ holds. This implies that that $\Gamma' T$ has a better objective value for \eqref{eqn:lrank} than $\Theta^*$ does. This contradicts the assumption that $\Theta^*$ is a solution of \eqref{eqn:lrank}; therefore, (b) holds.
\end{proof}
%% \begin{proof}[Proof of (c)]
%%   Directly follows from (b).
%% \end{proof}
Moreover, it directly follows from Theorem \ref{thm:equiv-2}(b) that the optimal objective value of \eqref{eqn:sid} and \eqref{eqn:lrank} are equal. Thus, the optimization problems \eqref{eqn:sid} and \eqref{eqn:lrank} are equivalent. 

\section{SID-DMD Algorithm}\label{sec:alg}
We exploit the equivalence results to derive a combined algorithm that we call SID-DMD. This algorithm aims to identify an order $n$ state-space model with observability index not greater than $s$ from a given sequence of output data $\{y_k\}_{k=i}^{j}$. The proposed algorithm takes the desired orders $n,s\in\mathbb{I}$ and output data sequence $\{y_k\}_{k=i}^{j}$ as inputs, and returns system matrices $(A,C)$. If needed, it also returns modes $(\Psi,\Lambda)$. The identification procedure is a combination of SID and DMD; in specific, the overall procedure involves three steps: (i) identification of a rank-$n$ extended AR mapping $\Theta$, (ii) extraction of the system matrices $(A,C)$ from AR mapping $\Theta$, (iii, optional) spatiotemporal mode decomposition to obtain the spatial and temporal modes $(\Psi,\Lambda)$. Notably, the SID-DMD algorithm provides an optimal model (solves \eqref{eqn:sid} and \eqref{eqn:lrank}) and can be computed using SVD (it is easy to implement).

\subsection{Rank-Constrained Matrix Regression}
The first step of the algorithm is the identification of a low-rank AR model \eqref{eqn:exar} by solving Problem \eqref{eqn:lrank}. The following proposition establishes a closed-form solution for this problem and highlights several properties.\\

\begin{proposition}\label{prop:lrank}
  The following holds:  \\
  (a) $\Theta^*=Z_{(n)} S_{2}^{-1} U_2^\top$ is a solution of \eqref{eqn:lrank}, where $Z=  \bY_f V_2$; $Z_{(n)}= U_1S_1V_1^\top$ is an $n$-truncated SVD of $Z$; $\bY_p=U_2S_2V_2^\top$ is an economic SVD.\footnote{Note that truncated SVDs are not necessarily unique due to the fact that the $n$th and $(n+1)$th largest singular values may be the same. As such, we consider $Z_{(n)}$ as a specific realization; so the mapping $\Theta^*$ is not necessarily unique. However, if the data contains noise, it is highly unlikely that the $n$th and $(n+1)$th largest singular values are equal.} \\
  (b) $\Theta$ is a solution of \eqref{eqn:lrank} if and only if $\Theta U_2(U_2)^\top = \Theta^* $ holds for some $\Theta^*$ and $\rank(\Theta)\leq n$.\\
  (c) Suppose that the $n$-truncated SVD of $Z$ is unique; then $\Theta^{*}$ is a unique solution of \eqref{eqn:lrank} if and only if $\bY_p$ has full row rank.\\
  (d) $\Vert \Theta^*\Vert_{F} \leq \Vert \Theta\Vert_{F}$ holds for any solution $\Theta$ of \eqref{eqn:lrank} if the $n$-truncated SVD of $Z$ is unique.\\
  (e) $\left\Vert \Theta^{\textrm{\normalfont full}} \bY_p - \Theta^* \bY_p \right\Vert_F =  (\sum_{k\geq n+1} (\sigma_k)^2 )^{1/2}$,  where $\Theta^{\text{\normalfont full}}:= Z S_2^{-1}U_2^\top$ is a full-rank solution of \eqref{eqn:lrank}; $\sigma_k$ is the $k$th largest singular value of $Z$.
\end{proposition}

\begin{proof}[Proof of (a)] By orthogonal invariance of the Frobenius norm, the squared objective of \eqref{eqn:lrank} can be rewritten as:
  \begin{align*}
    \Vert (\bY_f  -\Theta\bY_p)[V_2\;V_2^\perp]\Vert_F^2 
    =\Vert Z -\Theta  U_2 S_2 \Vert^2_F + \Vert\bY_f  V_2^\perp\Vert_F^2,
  \end{align*}
  where the columns of $V_2^\perp$ are the orthonormal basis of the orthogonal complement of the column space of $V_2$. The equality follows from the block structure. Observe that the second term is constant and thus it can be neglected. We observe that a lower bound of the objective under a rank constraint on $\Theta$ \eqref{eqn:lrank-con} can be found by a rank-$n$ approximation of $Z$. This is because $\rank(\Theta U_2 S_2)\leq n$ holds if $\rank(\Theta)\leq n$. This implies that $\Theta$ is a solution of \eqref{eqn:lrank} if
  \begin{align}\label{eqn:con-1}
    \Theta U_2 S_2 = Z_{(n)}
  \end{align} and $\rank(\Theta)\leq n$ holds for some $Z_{(n)}$ (this provides sufficient conditions for optimality). Now we check that $\Theta^*$ satisfies such conditions; one can see that $ \Theta^* U_2 S_2=Z_{(n)}$ holds. Lastly, one can see that $\rank(\Theta^*)\leq n$, since $\rank(Z_{(n)})\leq n$. Therefore, $\Theta^*$ satisfies the sufficient condition for optimality and is thus a solution.
\end{proof}

\begin{proof}[Proof of (b)] We first prove $\implies$; consider a solution $\Theta$ of \eqref{eqn:lrank}. Condition $\rank(\Theta)\leq n$ is satisfied from the feasibility. Since $\Theta^*$ is a solution, $\Theta$ should also satisfy $\Theta U_2 S_2 = Z_{(n)}$ for some $Z_{(n)}$ not to be worse than $\Theta^*$ in objective value. Finally, $\Theta U_2 S_2 = Z_{(n)}$ implies $\Theta U_2(U_2)^\top  = \Theta^* $ for some $\Theta^*$. We now prove $\impliedby$; above we saw that $\Theta U_2 S_2 = Z_{(n)}$ for some $Z_{(n)}$ and $\rank(\Theta)\leq n$ are sufficient conditions for optimality. Accordingly, it suffices to show that $\Theta U_2 S_2 = Z_{(n)}$ holds for some $Z_{(n)}$ if $\Theta U_2(U_2)^\top = \Theta^* $ for some $\Theta^*$. By right multiplying $U_2 S_2$ to $\Theta U_2(U_2)^\top = \Theta^* $, we have $\Theta U_2S_2= Z_{(n)}$. As such, \eqref{eqn:con-1} holds. 
\end{proof}

\begin{proof}[Proof of (c)] We first prove $\implies$ by contradiction; suppose that $\bY_p$ does not have full row rank. Let $\Theta^*=PQ^\top$, where $P,Q\in\mathbb{R}^{ms\times n}$ and $P$ has full column rank. We consider $\Theta' := \Theta^{*} + pq^\top$, where $p$ is a non-zero column of $P$ and $q$ is the first column of $Q^\perp$, where the columns of $Q^\perp$ are the orthonormal basis of the orthogonal complement of the column space of $Q$. Note that the column space of $Q^\perp$ is non-empty due to the assumption that $\bY_p$ does not have full row rank. Observe that $pq^\top\neq 0$. We have that $\Theta U_2(U_2)^\top = \Theta^* $ is satisfied (because $q^\top U_2=0$). Also, one can see that the columns of $\Theta^*$ span the column space of $\Theta'$; we thus have that $\rank(\Theta')\leq n$ is satisfied. These imply that $\Theta'\neq \Theta$ is a solution of \eqref{eqn:lrank}; thus, the solution of \eqref{eqn:lrank} is not unique. Therefore, the solution is unique only if $\bY_p$ has full row rank. We now prove $\impliedby$; we observe that, if $\bY_p$ has full row rank, $U_2(U_2)^\top=I$. Thus, the solution satisfies $\Theta=\Theta^*$. The uniqueness of $n$-truncated SVD of $Z$ eliminates the non-uniqueness of $\Theta^*$. Therefore, the solution is unique.
\end{proof}

\begin{proof}[Proof of (d)]
  Directly follows from (b).
\end{proof}

\begin{proof}[Proof of (e)] We have $(\Theta^{\text{full}} - \Theta^*)\bY_p = ( Z- Z_{(n)})( V_2)^\top$. By the orthogonal invariance of the Frobenius norm, $\Vert \Theta^{\text{full}} \bY_p - \Theta^* \bY_p \Vert_F= \Vert  Z- Z_{(n)}\Vert_F$. From Eckart-Young-Mirsky, $\Vert  Z- Z_{(n)}\Vert_F= (\sum_{k\geq n+1} (\sigma_k)^2 )^{1/2}$. This result holds regardless of the uniquenesss of $Z_{(n)}$.
\end{proof}

One can thus see that a solution of \eqref{eqn:lrank} can be computed by using SVDs (pseudoinverses can also be computed by SVDs). The solution of \eqref{eqn:lrank} is globally optimal and scalable; observe also from the solution form of \eqref{eqn:lrank} that one can naturally obtain the solution as a factorized form $\Theta=PQ^\top$, with $P,Q\in\mathbb{R}^{ms\times n}$ (e.g., $P=U_1$; $Q=U_2 S_2^{-1} V_1 S_1$). Retaining such a form is advantageous in that it requires less memory (for the case that $ms\gg n$). Note that Proposition \ref{prop:lrank} is stated in general terms; in other words, the results hold with arbitrary $\bY_p$ and $\bY_f$. Proposition \ref{prop:lrank} extends the results reported in the literature \cite{heas2017optimal,xiang2012optimal}. In partcular, (b) and (d) have not been reported in the literature (to the best of our knowledge).

One notable observation is that even with rank constraint, the sample size should be sufficiently large to guarantee the uniqueness of the solution (Proposition \ref{prop:lrank}(c)). If the sample size is insufficient (which is the case for most applications in DMD), $\Theta^*$ is a reasonable choice since it is the minimum norm solution.\\

\vspace{-0.2in}
\subsection{Estimation of System Matrices}
The second step of the SID-DMD algorithm is the estimation of the system matrices $(A,C)$ from the extended AR mapping $\Theta^*$. Here, we present a method that directly obtains system matrices from the factorized form of the solution $\Theta^*=PQ^\top$ of Problem \eqref{eqn:lrank}. The following proposition establishes this approach.\\

\begin{proposition}\label{prop:solb}
  Let $\Theta^*=PQ^\top$ with $P,Q\in\mathbb{R}^{ms\times n}$ be a solution of \eqref{eqn:lrank}. Then there exists a solution $(\Gamma^*,X^*)$ of \eqref{eqn:sid} such that $ A:= Q^\top  P$ and $ C:= P[1:m,:] $ solve:
  \begin{align}\label{eqn:ACprob}
    \min_{A,C}\;\left\Vert
    [X^*_f;Y_f]-[A;C]
      X^*_p\right\Vert_{F} ,
  \end{align}
  where $X^*_f:=X^*[:,2:\ell+1]$; $Y_f:=[y_{i+1}\;\cdots\;y_{j-s+1}]$.
\end{proposition}
\begin{proof}
\vspace{-0.2in}
  From Theorem \ref{thm:equiv-2}, $(P,Q^\top \bY)$ is a solution of \eqref{eqn:sid}. We choose $\Gamma^*=P$ and $X^*=Q^\top \bY$. From the optimality of $\Theta^*$ to \eqref{eqn:lrank}, we have that $P$ is a solution of $\min_{P'}\Vert \bY_f - P' X^*_p \Vert_F^2$, since $X^*_p=Q^\top \bY_p$. This implies that $\bY_f (\bX_p^*)^\dag= P (X^*_p)^\dag X^*_p$ (this follows from the full-rank case of Proposition \ref{prop:lrank} (b)). Now we left multiply $[Q^\top;I_{ms\times ms}[1:m,:]]$. This yields $ [X^*_f; Y_f](X^*_p)^\dag = [A;  C] X^*_p (X^*_p)^\dag$. Again, by Proposition \ref{prop:lrank}(b), one can see that $(A,C)$ is a solution of Problem \eqref{eqn:ACprob}.
\end{proof}

In other words, by extracting $(A,C)$ as $ A:= Q^\top  P$ and $ C:= P[1:m,:]$ from $\Theta^*=PQ^\top$, one can obtain $(A,C)$ that minimizes the prediction error of the state-space model (up to within a similarity transformation). 

\subsection{Mode Decomposition}
In high-dimensional systems, it is often useful to express the state evolution equation as an evolution of spatial and temporal modes. Here, we derive a similar representation with output predictions for \eqref{eqn:ss}.\\

\begin{proposition}
  For diagonalizable $A$, we have that $\widehat{y}_{k|\ell} = \Psi \Lambda^{k-\ell}\widehat{b}_{\ell|\ell}$, where $\widehat{x}_{k+1|\ell} := A \widehat{x}_{k|\ell}$; $\widehat{y}_{k|\ell} := C \widehat{x}_{k|\ell}$; $\widehat{b}_{k|\ell}:=\Phi^{-1}\widehat{x}_{k|\ell}$; an eigendecomposition of $A$ is $A\Phi = \Phi\Lambda $; $\Psi:=C\Phi$.
\end{proposition}
Proof directly follows from the definitions. We now redefine the modes in the context of output prediction of state-space model.\\

We define the columns $\{\psi_k\}_{k=1}^n$ of $\Psi$ as spatial modes and the diagonal components $\{\lambda_k\}_{k=1}^n$ of $\Lambda$ as temporal modes. The spatiotemporal modes can have complex components, which leads to oscillatory behavior. One can obtain the estimate $\widehat{x}_{\ell|\ell}$ of the current state $x_\ell$ by using observers, such as a Kalman Filter.

\subsection{SID-DMD Algorithm}
The results of this section are summarized in Algorithm \ref{alg:main}. The algorithm is divided into three subroutines: rank-constrained matrix regression, system matrices estimation, and mode decomposition. Here, {\tt svd}$(\cdot,${\tt 'econ'}$)$ is a function that computes economic SVD; {\tt svds}$(\cdot,n)$ is a function that computes $n$-truncated SVD; {\tt eig}$(\cdot)$ is a function that computes eigendecomposition. An implementation of this is available at \url{https://github.com/zavalab/JuliaBox/tree/master/SID_DMD}. 

\begin{algorithm}
  \caption{SID-DMD Algorithm}\label{alg:main}
  (a) Rank-Constrained Matrix Regression:
  
  \begin{algorithmic}[1]
    \REQUIRE $n,s\in\mathbb{I}_{>0}$, $\{ y_k\in\mathbb{R}^{m}\}_{k=i}^{j}$
    \STATE $\bY_{p} :=  [\by_{i}\;\cdots\; \by_{j-s}]$; $\bY_{f}:= [\by_{i+1}\;\cdots\; \by_{j-s+1}]$
    \STATE $(U_2,S_2,V_2)=${\tt svd}$( \bY_p,\text{\tt 'econ'})$
    \STATE $(U_1,S_1,V_1)=${\tt svds}$( \bY_f V_2,n)$
    \STATE $P=U_1$; $Q=U_2 S_2^{-1} V_1 S_1$
    \ENSURE $P\in\mathbb{R}^{ms\times n}$, $Q\in\mathbb{R}^{ms\times n}$
  \end{algorithmic}

  (b) System Matrices Estimation:
  
  \begin{algorithmic}[1]
    \REQUIRE $P\in\mathbb{R}^{ms \times n}$, $Q\in\mathbb{R}^{ms\times n}$
    \STATE $ A= Q^\top  P$; $ C= P[1:m,:]$
    \ENSURE ${ A}\in\mathbb{R}^{n\times n}$, ${ C}\in\mathbb{R}^{m\times n}$
  \end{algorithmic}

  (c) Mode Decomposition:
  
  \begin{algorithmic}[1]
    \REQUIRE $A\in\mathbb{R}^{n\times n}$, $C\in\mathbb{R}^{m\times n}$
    \STATE $(\Lambda,\Phi)= \tt{eig}(A)$
    \STATE $\Psi = C\Phi$
    \ENSURE ${ \Psi}\in\mathbb{R}^{m\times n}$, ${\Lambda}\in\mathbb{R}^{n\times n}$
  \end{algorithmic}
\end{algorithm}

\begin{figure}[!htp]
  \includegraphics[width=0.48\textwidth]{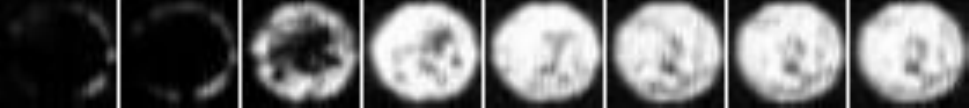}\\
  \includegraphics[width=0.48\textwidth]{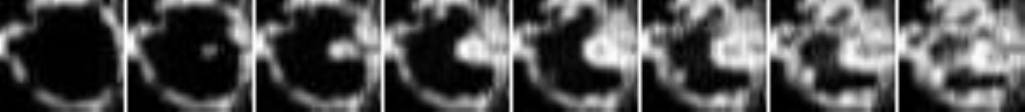}
  \caption{Snapshots of experimental data for DMMP (first row) and water (second row) at different times (left to right).}\label{fig:lc-raw}
  \vspace{0.1in}
  \includegraphics[width=0.49\textwidth]{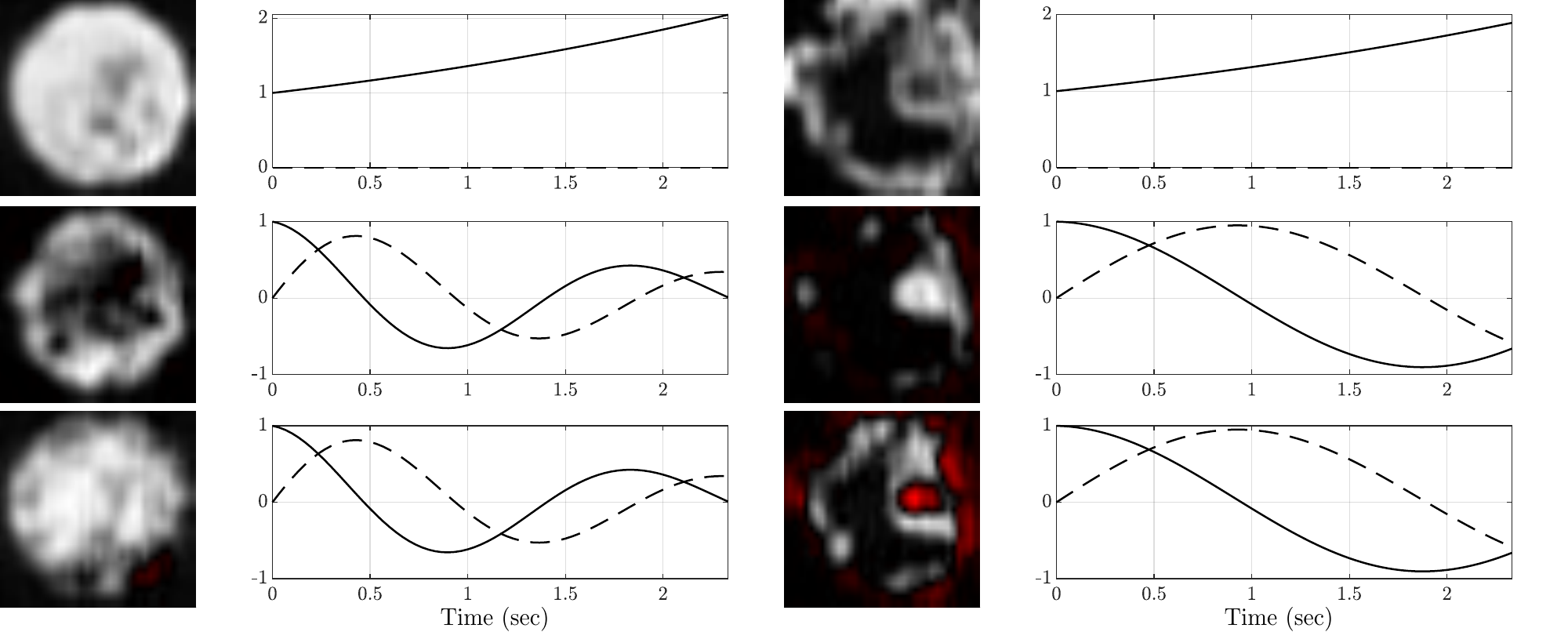}
  \caption{Spatiotemporal modes for LC data. First column: spatial modes of DMMP. Second column: temporal trends of DMMP. Third column: spatial modes of water. Fourth column: temporal trends of water. First, second, and third rows represent the real modes, the real parts of the complex modes, and the imaginary parts of the complex modes, respectively. The solid lines and the dashed lines represent the real part and the imaginary parts of the temporal trends.}\label{fig:lc-result}
\end{figure} 

\section{Case Study}\label{sec:num}
Liquid crystals (LCs) undergo surface-driven ordering transitions in the presence of chemical contaminants. The ordering transitions triggered by different gas-phase analytes produce distinct spatiotemporal (high-dimensional) patterns that can be detected by a polarizing microscope. This selectivity makes LCs flexible sensing platforms \cite{cao2018machine}. In this case study, we study spatiotemporal responses of LCs when exposed to different contaminants (DMMP and water). Spatial snapshots of the raw data are shown in Figure \ref{fig:lc-raw}. Initially, LC fields are black, and light intensity increases as the contaminant diffuses through the LC film. It is clear that DMMP and water produce distinct spatiotemporal patterns. The original data has dimension of $m=34\times 31$ for DMMP and $m=36\times 32$ for water. Both data have $71$ timeframes (snapshots). Algorithm \ref{alg:main} is used to identify the low-rank model, extract system matrices, and obtain the modes and we use $(n,s)=(3,20)$.

Modes for each case (DMMP and water) are visualized in Figure \ref{fig:lc-result}. We show the spatial modes as images, where pixels with positive intensities are colored in white and the pixels with negative intensities are colored in red. The temporal modes are visualized by their temporal trend, $(\lambda_k)^{t/\Delta t}$, where $\Delta t=1/30$ sec is the sampling time. Whenever the eigenpairs are obtained as conjugate pairs, we separately show the real and imaginary parts of the spatial modes. The associated real and imaginary parts of the temporal trends are plotted together. 
For DMMP and water, one real mode and two complex modes (conjugate pairs) are obtained. One can interpret the real dynamic modes as the slow-changing (growth/decay) mode and the complex modes as fast-changing (oscillatory) modes. We see that DMMP has a faster rate for the real mode, and a shorter period of oscillation for the complex modes. This implies that the transition occurs more rapidly with DMMP (this can be visually confirmed). Furthermore, we can see that the patterns of spatial modes are more uniform in DMMP. This implies that DMMP undergoes a uniform transition (the contaminant diffuses in a more homogenous manner). Furthermore, the structure of the spatial mode of DMMP indicates that the transition is started from the boundary and propagates towards the center.

\vspace{-0.1in}
\section{Conclusions}
We have presented equivalence theorems for subspace identification and dynamic mode decomposition and we have exploited these results to develop an algorithm that delivers a provably optimal model and is easy to implement. In future work, we will seek to establish equivalence for non-autonomous dynamic systems and we will seek to apply these data-driven capabilities to model predictive control.

\bibliographystyle{plain}      
\bibliography{low-rank}

\appendix
\section{SID Solves \eqref{eqn:sid}}\label{apx:sid}
We show that a well-known SID model known as UPC (a counterpart of N4SID for autonomous systems \cite{van2012subspace}), solves the optimization problem \eqref{eqn:sid}. Here, we use notation commonly used in SID literature: $A/B = A B^\dag B$ and $A/B^\perp = A (I-B^\dag B)$. UPC identifies $\Gamma$ as $\Gamma=US^{1/2}$, where $(\bY_f/\bY_p)_{(n)}=USV^\top$ is an $n$-truncated SVD of $(\bY_f/\bY_p)$. The following proposition shows that such $\Gamma$ solves \eqref{eqn:sid}.\\

\begin{proposition}\label{prop:n4sid}
  There exists $X\in\mathbb{R}^{n\times (\ell+1)}$ such that $\Gamma=US^{1/2}$ and $X$ are solutions of \eqref{eqn:sid}.
\end{proposition}
\begin{proof}
The objective function of \eqref{eqn:sid} can be rewritten as $\Vert \bY_f/\bY_p-\Gamma X_p\Vert_F^2 + \Vert \bY_f/\bY^\perp_p\Vert_F^2$. The separability of the Frobenius norm follows from the fact that the rows of $\bY_f/\bY_p^\perp$ are orthogonal to the rows of $\bY_f/\bY_p$ and $X_p$ (recall that the row space of $\bY$ should span that of $X$). One can also see that the second term is constant. Since $\rank(\Gamma X)\leq n$, the lower bound of the objective value is attained if $\Gamma X_p = (\bY_f/\bY_p)_{(n)}$. Here, we choose $X=S^{-1/2}U^\top \bY_f\bY_p^\dag\bY$. One can easily see that $\Gamma X_p = (\bY_f/\bY_p)_{(n)}$ and $\row(X)\subseteq \row(\bY)$ hold; therefore, $(\Gamma,X)$ is a solution of \eqref{eqn:sid}.
\end{proof}
\end{document}